\documentclass[12pt,leqno,letterpaper]{article}

\usepackage{amsmath,amsthm,enumerate,amssymb}
\usepackage[utf8]{inputenc}

\usepackage[english]{babel}
\usepackage{graphicx}
\usepackage{color}

\usepackage{graphicx}

\newtheorem{theorem}{Theorem}[section]

\newtheorem{lemma}[theorem]{Lemma}

\newtheorem{example}[theorem]{Example}

\newcommand{\tr}{{\rm Tr\hskip -0.2em}~}

\DeclareMathOperator{\frechetdiff}{\mathit d}
\newcommand{\fd}[1]{\hskip-0.2em\frechetdiff\hskip -0.3em{#1}}

\begin{document}

\title{A note on quantum entropy}
\author{Frank Hansen}
\date{November 30, 2015}

\maketitle

\begin{abstract} Incremental information, as measured by the quantum entropy, is increasing when two ensembles are united. This result was proved by Lieb and Ruskai, and it is the foundation for the proof of strong subadditivity of quantum entropy. We present a truly elementary proof of this fact in the context of the broader family of matrix entropies introduced by Chen and Tropp.
\end{abstract}

\section{Introduction}

Let $ \rho $ be a positive definite matrix on a bipartite system $ H=H_1\otimes H_2 $ of Hilbert spaces $ H_1 $ and $ H_2 $ of finite dimensions. Lieb and Ruskai \cite[Theorem 1]{kn:lieb:1973:3} proved that the function
\[
\rho\to S(\rho)-S(\rho_1)
\]
is concave in positive definite $ \rho $ on $ H, $ where $ \rho_1=\tr_2\rho $ is the partial trace of $ \rho $ on $ H_1 $ and $ S(\rho)=-\tr\rho\log\rho $ is the quantum entropy of $ \rho. $ The proof used Klein's inequality and Lieb's concavity theorem. Before giving a truly elementary proof of this result we broaden the investigation and consider functions of the form
\begin{equation}\label{general difference map}
G(\rho)=d_2^{-1}\tr_{12}f(d_2\rho)-\tr_1 f(\rho_1),
\end{equation}
where $ f\colon(0,\infty)\to\mathbf R $ is a given function, and $ d_2 $ is the dimension of $ H_2. $ 
If $ f $ is sufficiently smooth then $ G $ is Fréchet differentiable and the first Fréchet differential is given by
\[
\begin{array}{rl}
dG(\rho)h&=\tr_{12}\fd{}f(d_2\rho)h - \tr_1\fd{}f(\rho_1)h_1\\[1.5ex]
&=\tr_{12} f'(d_2\rho)h - \tr_1 f'(\rho_1)h_1,
\end{array}
\]
where we used $ \tr\fd{}f(x)h=\tr f'(x)h, $ see for example \cite[Theorem 2.2]{kn:hansen:1995}. We continue to calculate the second Fréchet differential
\begin{equation}\label{second frechet differential}
d^2G(\rho)(h,h)=d_2\tr_{12} h\fd{}f'(d_2\rho)h - \tr_1 h_1\fd{}f'(\rho_1)h_1
\end{equation}
in positive definite $ \rho $ and self-adjoint $ h. $ 

\subsection{Matrix entropies}

Matrix entropies were introduced by Chen and Tropp  \cite{kn:tropp:2014} as a tool to obtain concentration inequalities for random matrices, and their representing functions may be characterised in various ways  \cite[Theorem 1.2]{kn:hansen:2014:3}. In particular, a twice differentiable function $ f\colon(0,\infty)\to\mathbf R $ is (the representing function) of a matrix entropy if and only if the function of two variables
\begin{equation}\label{equivalent convexity condition}
(x,h)\to\tr h^*\fd{}f'(x)h
\end{equation}
is convex for positive definite $ \rho $ and arbitrary $ h. $ Lieb  \cite[Theorem 3]{kn:lieb:1973:1} proved that the function
\begin{equation}\label{convexity statement for the logarithm}
(x,h)\to\tr h^*\fd{}\log(x)h=\int_0^{\infty}\tr h^*\frac{1}{\rho+\lambda}h\frac{1}{\rho+\lambda}\,d\lambda
\end{equation}
is convex. The function $ t\to t\log t $ is therefore a matrix entropy (this result was obtained by different means in \cite{kn:tropp:2014}).

The convexity statement in (\ref{convexity statement for the logarithm}) may be easily obtained by the following argument. Consider the positive function
\[
k(t,s)=\frac{\log t-\log s}{t-s}=\int_0^1 (\lambda t+(1-\lambda)s)^{-1}\,d\lambda\qquad t,s>0,
\]
and let $L_x$ and $R_x$ denote left and right multiplication with $x.$ Then
\[
\tr h^*\fd{}\log(x) h=\sum_{i,j=1}^n |(he_i\mid e_j)|^2 \frac{\log\lambda_i-\log\lambda_j}{\lambda_i-\lambda_j}=\tr h^*\, k(L_x,R_x) h,
\]
where the intermediary calculation is carried out in an orthonormal basis $ (e_1,\dots,e_n) $ of eigenvectors of $ x $ with corresponding eigenvalues $ \lambda_1,\dots,\lambda_n $ counted with multiplicity. Combining the two formulas we obtain
\[\label{transform equation between derivatives}
\tr h^*\fd{}\log(x) h=\int_0^1 \tr h^* (\lambda L_x + (1-\lambda) R_x)^{-1} h\,d\lambda.
\]
The convexity in (\ref{convexity statement for the logarithm}) now follows from convexity of the operator map $ (A,B)\to B^*A^{-1}B, $ where $ A $ is positive definite, proved by Lieb and Ruskai \cite{kn:lieb:1974}. With Ando's elegant proof \cite[Theorem 1]{kn:ando:1979} this result is readily accessible.

The above line of arguments may be generalised, and we obtain that a sufficient condition for  a function $ f\colon(0,\infty)\to\mathbf R $ to be a matrix entropy is that the second derivative $ f'' $ is positive, decreasing and operator convex \cite[Theorem 1.3]{kn:hansen:2014:3}.

\begin{lemma}\label{Main lemma}
Let $ \Phi $ be a conditional expectation, and let $ f $ be a matrix entropy.
Then
\[
\tr \Phi(h)^*\fd{}f'\bigl(\Phi(x)\bigr)\Phi(h)\le \tr h^*\fd{}f'(x)h
\]
for positive definite $ \rho $ and arbitrary $ h. $
\end{lemma}

\begin{proof}
In the finite dimensional case that we consider $ \Phi $ may be written on the form
\[
\Phi(x)=\sum_{i=1}^n p_i u_i^* x u_i
\]
for unitaries $ u_1,\dots,u_n $ and non-negative weights $ p_1,\dots,p_n $ summing up to one. The
convexity of the map in (\ref{equivalent convexity condition}) thus yields
\[
\begin{array}{rl}
\tr \Phi(h)^*\fd{}f'\bigl(\Phi(x)\bigr)\Phi(h)&\displaystyle\le\sum_{i=1}^n p_i \tr u_i^*hu_i\fd{}f'(u_i^*xu_i)u_i^*hu_i\\[3.5ex]
&= \tr h^*\fd{}f'(x)h,
\end{array}
\]
which is the assertion.
\end{proof}

\section{The main result}

\begin{theorem}
Let $ f\colon(0,\infty)\to\mathbf R $ be a matrix entropy. The function
\[
G(\rho)=d_2^{-1}\tr_{12} f(d_2\rho)-\tr_1 f(\rho_1)
\]
is convex in positive definite $ \rho $ acting on a bipartite system $ H=H_1\otimes H_2 $ of Hilbert spaces of finite dimensions, where $ \rho_1 $ denotes the partial trace of $ \rho $ on $ H_1 $ and $ d_2=\dim H_2. $ 
\end{theorem}

\begin{proof} 
We may write $ \rho_1\otimes 1_2=d_2\pi_1(\rho) $
in terms of a conditional expectation $ \pi_1 $ on $ B(H). $ Since $ f $ is a matrix entropy we obtain
\[
\begin{array}{l}
\tr_1 h_1^*\fd{}f'(\rho_1)h_1
=d_2^{-1}\tr_{12}(h_1\otimes 1_2)^*\fd{}f'(\rho_1\otimes 1_2)(h_1\otimes 1_2)\\[1.5ex]
=d_2 \tr_{12} \pi_1(h)^* \fd{}f'\bigl(d_2\pi_1(\rho)\bigr)\pi_1(h)\\[1.5ex]
\le d_2\tr_{12} h^*\fd{}f'(d_2\rho)h,
\end{array}
\]
where we used Lemma~\ref{Main lemma}.
It then follows from (\ref{second frechet differential}) that the second Fréchet differential $ d^2 G(\rho)(h,h) $ is non-negative, so $ G $ is convex.
\end{proof}

\begin{example}
If we consider the matrix entropy $ f(t)=t\log t $ then the map
\[
\begin{array}{rl}
G(\rho)&=d_2^{-1}\tr_{12}d_2\rho\log(d_2\rho)-\tr_1 \rho_1\log\rho_1\\[1.5ex]
&=\log d_2\tr_{12}\rho +\tr_{12}\rho\log\rho-\tr_1\rho_1\log\rho_1
\end{array}
\]
is convex. But this shows that the map $ \rho\to S(\rho) - S(\rho_1) $
is concave, where $ S(\rho)=-\tr\rho\log\rho $ is the von Neumann entropy. 
\end{example}

\begin{example}
If we consider the matrix entropy $ f(t)=t^p $ for $ 1\le p\le 2, $ it follows that the map
\[
G(\rho)=d_2^{p-1}\tr_{12}\rho^p-\tr_1\rho_1^p
\]
is convex.
\end{example}

{\small


\vfill

\noindent Frank Hansen: Institute for Excellence in Higher Education, Tohoku University, Japan.\\
Email: frank.hansen@m.tohoku.ac.jp.
      }

\end{document}